\documentclass[preprint,12pt]{elsarticle}

\usepackage{amssymb}
\usepackage{amsmath}
\usepackage{amsthm}
\usepackage[hidelinks]{hyperref}

\newtheorem{theorem}{Theorem}
\newtheorem{lemma}{Lemma}
\newtheorem{corollary}{Corollary}
\theoremstyle{definition}
\newtheorem{algorithm}{Algorithm}
\newcommand{\OPT}{\mathrm{OPT}}
\newcommand{\diam}{\operatorname{diam}}

\journal{Theoretical Computer Science}

\begin{document}
\begin{frontmatter}

\title{Approximation Algorithms for Emergency Vertex Cover with Few Stations and Fault-Tolerant Coverage}

\author[inst1]{Quanlong Niu}

\author[inst2]{Riheng Jia\corref{cor1}}
\ead{rihengjia@zjnu.edu.cn}

\cortext[cor1]{Corresponding author}

\affiliation[inst1]{
    organization={School of Big Data and Artificial Intelligence, Fuyang University of Technology},
    city={Fuyang},
    postcode={236000},
    state={Anhui},
    country={China}
}

\affiliation[inst2]{
    organization={School of Computer Science and Technology, Zhejiang Normal University},
    city={Jinhua},
    postcode={321004},
    state={Zhejiang},
    country={China}
}

\begin{abstract}
The Emergency Vertex Cover problem (Em-VC) assigns budgets to stations at graph vertices to cover every edge while minimizing the total budget. The cost of covering an edge from a station is the shortest-path distance from the station to the nearer endpoint plus the edge length. We study how restricting the number of stations affects the cost of ordinary and fault-tolerant coverage. In ordinary Em-VC, we prove that the optimal solution using at most two stations costs at most $5/4$ times the unrestricted optimum, improving the known $4/3$ guarantee for a single station. A recurrence extends this bound to larger station limits, giving approximation factors that converge to $\sqrt{5}-1$. Budget-threshold enumeration computes the restricted optimum in polynomial time for every fixed station limit. Under $k$-fold coverage, every edge must be covered by at least $k$ distinct stations, ensuring coverage after any $k-1$ station failures. We prove that selecting the $k$ vertices with the smallest full-coverage costs yields a $2$-approximation relative to the unrestricted optimum for $k$-fold coverage. When $k=2$, we combine upper bounds derived from the structure of the covering balls to improve the approximation guarantee to $5/3$.
\end{abstract}

\begin{keyword}
Emergency vertex cover \sep Coverage with few stations \sep Fault-tolerant coverage \sep Approximation algorithms \sep Expanded metric graph
\end{keyword}
\end{frontmatter}

\section{Introduction}
\label{sec:intro}

Emergency services such as firefighting and emergency medical care depend on an appropriate placement of facilities. Set cover provides a natural model for such problems: each candidate facility corresponds to a set of demand points it can serve, and facilities are selected to satisfy the coverage requirements~\cite{toregas1971}. When service ranges depend on allocated resources, the coverage sets also vary, so facility placement and resource allocation must be considered together. For example, broadcast domination assigns broadcast strengths to graph vertices to determine their coverage ranges and minimizes the sum of these strengths~\cite{heggernes2006}. The objects to be covered may also be graph edges: power vertex cover requires every edge to have an endpoint whose power is at least the edge weight~\cite{angel2015}. In emergency response, a team dispatched from a station may travel through the network and fully traverse an edge not incident to the station. The required budget then depends on both the edge length and the distance from the station to the target edge.

To model this form of response, Angel et al.~\cite{angel2026} introduced Emergency Vertex Cover (Em-VC). Stations are placed at graph vertices, and the response cost from a station to an edge is the shortest-path distance to the nearer endpoint plus the edge length. Unlike power vertex cover, where an edge can only be covered by one of its endpoints, Em-VC allows a station at any vertex to cover the edge, provided that its budget meets the response cost. Every edge must be covered by at least one station, and the sum of the station budgets is minimized. Angel et al. proved that the problem is NP-hard and gave a $4/3$-approximation algorithm using a single station. They also gave a polynomial-time exact algorithm for the continuous version, which allows stations in edge interiors, and studied special cases including unit-length graphs, complete graphs, and trees.

For ordinary Em-VC, the existing $4/3$-approximation algorithm uses only one station. We investigate the approximation guarantees obtainable when a small number of stations is allowed. Building on the expanded-graph construction and exclusive-edge decomposition of Angel et al.~\cite{angel2026}, we combine single-station cost estimates with budget relations among the residual components. By constructing and comparing three candidate solutions based on concentrating coverage at one station, absorbing components, and recursively covering a component, we derive a recurrence relating the station limit to the approximation guarantee. We prove that the optimal solution using at most two stations costs at most $5/4$ times the unrestricted optimum, improving the existing $4/3$ guarantee to $5/4$. The upper bounds given by the recurrence decrease to $\sqrt{5}-1$ as the station limit increases. We also give an exact algorithm for the problem with at most two stations. For every fixed station limit, budget-threshold enumeration computes an optimal restricted solution in polynomial time, yielding an algorithm with the corresponding approximation guarantee.

Emergency response also requires coverage to be maintained when stations fail. Backup coverage models improve service availability by covering demand points more than once~\cite{hogan1986}. In Em-VC, if every edge is covered by at least $k$ distinct stations, it remains covered after any $k-1$ station failures. We study this $k$-fold coverage requirement and consider solutions using exactly $k$ stations. Each station must then cover every edge independently, so the optimal such solution consists of the $k$ vertices with the smallest full-coverage costs. A coverage-length estimate along shortest paths in the expanded graph gives a $2$-approximation guarantee relative to the unrestricted optimum for $k$-fold coverage. For twofold coverage, we use the maximum weight of an independent set in the ball intersection graph and the third-largest station budget as auxiliary quantities. Combining these quantities with savings from merging balls that cover a central edge gives three complementary upper bounds, whose weighted combination improves the approximation guarantee to $5/3$.

\section{Related work}
\label{sec:related}

Covering models for emergency facility location have long addressed service standards, facility limits, and backup coverage. Toregas et al.~\cite{toregas1971} formulated the location set-covering problem, which minimizes the number of facilities needed to serve all demand points within a prescribed time or distance. Church and ReVelle~\cite{church1974} studied the maximal covering location problem, maximizing the covered population when the number of facilities is limited. Hogan and ReVelle~\cite{hogan1986} introduced backup coverage into these models, examining the trade-off between first and second coverage and discussing higher coverage multiplicities. The service threshold is generally fixed in these formulations; in Em-VC, the response budgets are chosen as part of the solution.

Power vertex cover extends vertex selection to power allocation while retaining coverage by edge endpoints. Angel et al.~\cite{angel2015} gave a $2$-approximation algorithm for the problem, in which an edge is covered if at least one endpoint has power no smaller than its weight. Angel et al.~\cite{angel2018} subsequently studied its parameterized complexity. They obtained fixed-parameter algorithms parameterized by the optimal total power or by the number of vertices assigned positive power, established hardness with respect to treewidth, and gave an approximation scheme parameterized by treewidth and the accuracy parameter.

Shortest-path distances also determine coverage in broadcast domination. This problem minimizes the sum of broadcast strengths assigned to vertices, with every vertex required to lie within the shortest-path range of a broadcast source. Heggernes and Lokshtanov~\cite{heggernes2006} gave an exact $O(n^6)$-time algorithm on connected unweighted graphs, where $n$ is the number of vertices. Papadopoulos~\cite{papadopoulos2026} improved the running time to $O(n^5)$. Em-VC shares the use of shortest-path distances, but covering the vertices does not suffice: each edge must be reached and traversed in full by a response from a single station. The model and expanded-graph construction of Angel et al.~\cite{angel2026} express this additional requirement through response costs and auxiliary leaves.

The minimum sum of radii problem concerns a related budget objective with a prescribed number of centers. It selects at most that many balls to cover demand points and minimizes the sum of their radii. Charikar and Panigrahy~\cite{charikar2004} gave a constant-factor approximation using a primal--dual method. For graph shortest-path metrics, Drexler et al.~\cite{drexler2023} studied more general radius-dependent costs and obtained polynomial-time algorithms on graphs of bounded treewidth. Buchem et al.~\cite{buchem2024} obtained a $(3+\varepsilon)$-approximation for the minimum sum of radii problem and its version with outliers. In these approximation results, the solution and the optimum are subject to the same limit on the number of centers. Our finite-station bounds use a different benchmark: the optimum with no station limit.

For multiple coverage, Bhowmick et al.~\cite{bhowmick2021} studied metric multicover, where each demand point must lie in balls centered at several distinct facilities and the objective is to minimize the sum of powers of the radii. They obtained constant-factor approximation algorithms by reducing multicover to single-cover instances, and considered nonuniform coverage requirements and a limit on the number of facilities. Through the auxiliary leaves of the expanded graph, $k$-fold edge coverage in Em-VC has a corresponding ball formulation. The question studied here is how the cost of using exactly $k$ stations, the minimum feasible number, compares with the unrestricted optimum for the same coverage requirement.
\section{Problem formulation}
\label{sec:model}

\subsection{Response costs and ordinary coverage}
\label{sec:ordinarymodel}

Let $G=(V,E)$ be a finite, connected, undirected simple graph, with $n=|V|$, $m=|E|$, and $E\ne\varnothing$. Each edge $e\in E$ has a positive length $w(e)$, and $d(u,v)$ denotes the shortest-path distance between vertices $u$ and $v$ in the original graph. Emergency stations may only be placed at original vertices, with at most one station per vertex. A budget $p_v\ge 0$ is assigned to each vertex $v$. Vertices with zero budget are not counted as active stations.

For an edge $e=\{r,s\}$, the response cost from a station at $v$ is the shortest-path distance to the nearer endpoint plus the edge length:
\[
c_v(e)=\min\{d(v,r),d(v,s)\}+w(e).
\]
Station $v$ covers edge $e$ if and only if $p_v\ge c_v(e)$. The minimum budget required for a station to cover a nonempty set of edges is the maximum of their response costs. Ordinary Em-VC requires every edge to be covered by at least one station:
\[
\bigl|\{v\in V:p_v\ge c_v(e)\}\bigr|\ge 1,
\qquad e\in E.
\]
The objective is to minimize the total budget
\[
S(p)=\sum_{v\in V}p_v.
\]
Let $\OPT$ denote the unrestricted optimum, and let $A_\ell$ denote the optimum using at most $\ell$ positive-budget stations. Then
\[
\OPT\le A_{\ell+1}\le A_\ell,
\qquad \ell\ge 1.
\]

For a fixed vertex $v$, its set of covered edges changes only when its budget reaches a response threshold. Thus, its budget can be reduced to a value in
\[
D_v=\{0\}\cup\{c_v(e):e\in E\}
\]
without changing the set of edges it covers. Each $D_v$ contains at most $m+1$ values. Consequently, ordinary Em-VC and its restriction to at most $\ell$ stations each admit an optimal solution in which every budget $p_v$ belongs to $D_v$.

\subsection{Multiple coverage and solutions with the minimum number of stations}
\label{sec:robustmodel}

For an integer $1\le k\le n$, $k$-Robust Em-VC requires every edge to be covered by at least $k$ distinct stations:
\[
\bigl|\{v\in V:p_v\ge c_v(e)\}\bigr|\ge k,
\qquad e\in E.
\]
Equivalently, every edge remains covered by at least one surviving station after the failure of any set of at most $k-1$ stations. Let $\OPT_k$ denote the unrestricted optimum for this problem; in particular, $\OPT_1=\OPT$.

Define the full-coverage cost of vertex $v$ by
\[
M(v)=\max_{e\in E}c_v(e).
\]
If exactly $k$ stations provide $k$-fold coverage, each must cover every edge. Order the vertices as $v_1,\ldots,v_n$ so that $M(v_1)\le\cdots\le M(v_n)$. The optimal cost using exactly $k$ stations is then
\[
U_k=\sum_{i=1}^k M(v_i).
\]
In particular, $U_1=A_1$. Note that $A_2$ is the optimum for ordinary coverage using at most two stations, whereas $U_2$ is the optimum for twofold coverage using exactly two stations.

\section{The expanded metric graph and basic properties}
\label{sec:expanded}

We use the expanded-graph construction of Angel et al.~\cite{angel2026} to express edge coverage as ball coverage. We then establish the path-length, ball-merging, and center-rounding bounds needed below, following the geometric approach of that work.

\subsection{From edge coverage to ball coverage}
\label{sec:ballmodel}

For each original edge $e=\{r,s\}$, introduce two new points $h_e$ and $t_e$. Replace the edge by the path $r-h_e-s$ and add the pendant edge $h_et_e$. Each of the three resulting edges has length $w(e)/2$. Regard each edge of the resulting weighted graph as a continuous line segment and equip it with the shortest-path metric. This gives a compact metric graph $X$, with distance $d_X$. The union of all paths $r-h_e-s$ is called the core. Write $B_X(x,r)$ for the closed ball in $X$ centered at $x$ with radius $r$.

\begin{lemma}[Equivalent representation of edge coverage]
\label{lem:equivalence}
For every original vertex $v$ and edge $e$,
\[
d_X(v,t_e)=c_v(e).
\]
Moreover, if a ball $B_X(v,p_v)$ centered at an original vertex contains $t_e$, then it contains all three line segments corresponding to $e$. Hence, $k$-fold edge coverage in the original graph is equivalent to $k$-fold coverage of all of $X$ by these balls.
\end{lemma}

\begin{proof}
Subdividing edges preserves distances between original vertices, and pendant edges do not shorten these distances. Every path from $v$ to $t_e$ passes through $h_e$, so
\[
d_X(v,t_e)=\min\{d(v,r),d(v,s)\}
              +\frac{w(e)}2+\frac{w(e)}2=c_v(e).
\]
Let $r$ be an endpoint nearer to $v$. Every point on $r-h_e-s$ or $h_et_e$ is reachable from $r$ by a path of length at most $w(e)$. Thus, if $p_v\ge c_v(e)$, all three corresponding line segments belong to $B_X(v,p_v)$. Each of the $k$ stations covering an edge therefore covers its entire gadget, so every point of $X$ lies in at least $k$ distinct balls. Conversely, $k$-fold coverage of each leaf $t_e$ directly provides $k$ distinct stations covering edge $e$.
\end{proof}

It follows that, for every original vertex $v$, the maximum distance from $v$ to a point of $X$ equals its full-coverage cost:
\[
\max_{x\in X}d_X(v,x)=\max_{e\in E}d_X(v,t_e)=M(v).
\]

\subsection{Coverage length along shortest paths}
\label{sec:path}

\begin{lemma}[Path-length bound]
\label{lem:path}
Let $Q$ be a shortest path in $X$. The total length of $Q$ covered by a ball $B_X(v,p_v)$ is at most $2p_v$. If $Q$ starts at $v$, this length is at most $p_v$.
\end{lemma}

\begin{proof}
Parameterize $Q$ by arc length. For any two points $x,y\in Q\cap B_X(v,p_v)$, the subpath between them is shortest. The absolute difference between their arc-length parameters therefore equals $d_X(x,y)$, and
\[
d_X(x,y)\le d_X(x,v)+d_X(v,y)\le 2p_v.
\]
Hence, the intersection of the ball with $Q$ lies in a parameter interval of length at most $2p_v$, which also bounds its total covered length. If $Q$ starts at $v$, the path parameter equals the distance from $v$, so the ball covers only the portion with parameter at most $p_v$.
\end{proof}

\begin{corollary}
\label{cor:path}
Consider the station balls $B_X(v,p_v)$, $v\in V$, with total budget $S=\sum_{v\in V}p_v$. If these balls cover $X$, then $\diam(X)\le 2S$. If they provide $k$-fold coverage, then, for every original vertex $v$,
\[
kM(v)\le 2S-p_v.
\]
\end{corollary}

\begin{proof}
Take a shortest path between any two points and sum the covered lengths contributed by all balls. Lemma~\ref{lem:path} bounds the path length by $2S$. For the second claim, take a shortest path from $v$ to a farthest point, of length $M(v)$. Integrating the coverage multiplicity along this path gives a total covered length of at least $kM(v)$. The ball centered at $v$ contributes at most $p_v$, and every other ball contributes at most twice its budget. Thus,
\[
kM(v)\le p_v+\sum_{u\ne v}2p_u=2S-p_v.
\]
The same inequality holds when $p_v=0$.
\end{proof}

\subsection{Merging intersecting balls and continuous centers}
\label{sec:merging}

\begin{lemma}[Merging intersecting balls]
\label{lem:merge}
Suppose two balls have radii $r,s$ and their centers are at distance $d\le r+s$. Their union is contained in a ball of radius at most
\[
R=\max\left\{r,s,\frac{r+s+d}{2}\right\}\le r+s.
\]
\end{lemma}

\begin{proof}
If $r\ge d+s$, the first ball contains the second; the case $s\ge d+r$ is symmetric. Otherwise, choose a point on a shortest path between the centers at distance $(d+s-r)/2$ from the first center. By the triangle inequality, the ball centered at this point with radius $(r+s+d)/2$ contains both original balls.
\end{proof}

Define the continuous center radius of $X$ as
\[
C=\min_{z\in X}\max_{x\in X}d_X(z,x).
\]
Compactness of $X$ ensures that the minimum is attained. If finitely many closed balls cover the connected space $X$, their intersection graph is connected. Otherwise, grouping the balls by different connected components of that graph would partition $X$ into two disjoint nonempty closed sets. Repeatedly merge intersecting balls. Each merged ball contains both original balls, so the process eventually produces one ball covering $X$. Lemma~\ref{lem:merge} shows that, if the original total budget is $S$, then
\[
C\le S.
\]

\begin{lemma}[Rounding a continuous center]
\label{lem:round}
A continuous center lies in the core. Suppose such a center belongs to the core path corresponding to an original edge $e_0=\{r,s\}$, and let $w=w(e_0)$. Then
\[
A_1\le C+\frac w2,
\qquad U_2\le 2C+w.
\]
\end{lemma}

\begin{proof}
First consider a pendant edge $h_et_e$. Let $D$ be the maximum distance from $h_e$ to points outside this pendant branch, retaining the attachment point $h_e$. The endpoints of the original edge are at distance $w(e)/2$ from $h_e$, so $D\ge w(e)/2$. After moving a distance $a>0$ from $h_e$ along the pendant edge, the maximum distance to any point of $X$ is
\[
\max\left\{D+a,\frac{w(e)}2-a\right\}=D+a>D.
\]
Thus, no continuous center lies in a pendant-edge interior or at its leaf.

Let $z$ be a continuous center on the core path corresponding to $e_0=\{r,s\}$. Let $x$ be the length of the subpath from $r$ to $z$; the subpath from $z$ to $s$ then has length $w-x$. The triangle inequality gives
\[
M(r)\le C+x,
\qquad M(s)\le C+w-x.
\]
Choosing the endpoint with smaller full-coverage cost yields $A_1\le C+w/2$. Giving each endpoint enough budget to cover the whole graph yields $U_2\le 2C+w$. If the center is an original vertex, choose any incident edge and set $x=0$; the same argument applies.
\end{proof}

\section{Approximation for ordinary coverage with few stations}
\label{sec:finite}

\subsection{Two upper bounds on the single-station cost}

\begin{lemma}[Single-station cost estimate]
\label{lem:single}
For any feasible ordinary-coverage solution with total budget $S$ and largest station budget $P$,
\[
A_1\le S+\frac P2,
\qquad A_1\le 2S-P.
\]
Consequently, we recover the single-station approximation guarantee $A_1\le \frac43\OPT$ of~\cite{angel2026}.
\end{lemma}

\begin{proof}
Let $w$ be the length of an original edge whose corresponding core path contains a continuous center. This edge is covered by a station with budget at most $P$, so $w\le P$. Combining $C\le S$ with Lemma~\ref{lem:round} gives
\[
A_1\le C+\frac w2\le S+\frac P2.
\]
For the other bound, choose a station $u$ with budget $P$. Corollary~\ref{cor:path} with $k=1$ gives
\[
A_1\le M(u)\le 2S-P.
\]
Weighting the two bounds by $2$ and $1$, respectively, yields $3A_1\le 4S$. Taking $S=\OPT$ proves the approximation guarantee.
\end{proof}

\subsection{Exclusive edges and residual components}
\label{sec:residual}

We use the exclusive-edge decomposition of Angel et al.~\cite{angel2026} to bound the cost of covering each residual component. Fix an optimal ordinary Em-VC solution with total budget $S=\OPT$. Let $u$ be a largest-budget station, with budget $P$. Define its exclusive-edge set by
\[
E_u=\{e\in E:p_u\ge c_u(e),\;
                  p_v<c_v(e)\text{ for every }v\ne u\}.
\]
After deleting $E_u$, let $H_1,\ldots,H_q$ be the connected components of the remaining graph, including isolated vertices. For each $H_j$, construct its local expanded metric graph $X_j$ as in Section~\ref{sec:ballmodel}, and write $d_{X_j}$ for its internal shortest-path distance. If $H_j$ has no edges, $X_j$ consists of a single point.

\begin{lemma}[Local coverage of residual components]
\label{lem:residual}
For each $H_j$, the original stations in $H_j$ other than $u$, with their budgets unchanged, cover all its edges using internal distances. Let
\[
R_j=\sum_{v\in V(H_j)\setminus\{u\}}p_v.
\]
Then $\diam(X_j)\le 2R_j$. Moreover, increasing the budget of $u$ to $P+2R_j$ allows it to cover every edge in $H_j$.
\end{lemma}

\begin{proof}
Consider an edge $e$ that was not deleted. It has a covering station $v\ne u$. Take a shortest path from $v$ to a nearer endpoint of $e$ and then traverse $e$ in full. The resulting service walk has length at most $p_v$. If it traversed an edge $f\in E_u$, the prefix ending immediately after the traversal of $f$ would also have length at most $p_v$. Station $v$ would then cover $f$, contradicting its exclusivity. The service walk therefore lies entirely in the residual graph. Thus, $v$ and $e$ lie in the same component, and the internal response cost is still at most $p_v$. This proves local coverage; Corollary~\ref{cor:path} gives the local diameter bound. A single-point component has diameter zero, so the bound also holds there.

To prove the absorption claim, if $u\in V(H_j)$, take $z=u$ as the access point. Otherwise, connectivity of the original graph ensures that $H_j$ is incident to a deleted edge. Let $z$ be its endpoint in $H_j$. That edge is covered by $u$, so $d_X(u,z)\le P$. For every local leaf $t_e\in X_j$,
\[
c_u(e)=d_X(u,t_e)
 \le d_X(u,z)+d_{X_j}(z,t_e)
 \le P+2R_j.
\]
Hence, budget $P+2R_j$ suffices for $u$ to absorb this component.
\end{proof}

\subsection{A recurrence for the approximation guarantee}

\begin{theorem}[Approximation with few stations]
\label{thm:finite}
Define
\[
\alpha_1=\frac43,
\qquad
\alpha_\ell=1+\frac{\alpha_{\ell-1}}{4+\alpha_{\ell-1}},
\qquad \ell\ge 2.
\]
For every positive integer $\ell$,
\[
A_\ell\le \alpha_\ell\OPT.
\]
\end{theorem}

\begin{proof}
We use induction on $\ell$. The case $\ell=1$ follows from Lemma~\ref{lem:single}. Let $\ell\ge 2$, assume the result for $\ell-1$, and write $a=\alpha_{\ell-1}$. Use the optimal solution and component decomposition from Section~\ref{sec:residual}, and set
\[
x=\max_j R_j,
\qquad y=\sum_j R_j-x.
\]
Since the components partition all vertices, $S=P+x+y$. We construct three candidate solutions using at most $\ell$ stations.

The first is an optimal single-station solution. Lemma~\ref{lem:single} gives
\[
A_\ell-S\le A_1-S\le \frac P2.
\]
The second assigns every residual component to $u$. Since $R_j\le x$ for all $j$, Lemma~\ref{lem:residual} shows that budget $P+2x$ suffices, while the original exclusive edges remain covered. Thus,
\[
A_\ell-S\le P+2x-S=x-y.
\]
For the third solution, retain a component with budget sum $x$ and apply the induction hypothesis to its local instance. The existing feasible solution for this component has total cost $x$, so its optimum is at most $x$. By the induction hypothesis, at most $\ell-1$ stations can cover it at cost at most $ax$. If it has no edges, no stations are needed. Assign the other components to $u$. Their individual budget sums are at most $y$, so budget $P+2y$ suffices. Using the same stations and budgets in the original graph preserves feasibility, since shortest-path distances cannot increase. If the local solution also places a station at $u$, combine its two budgets by taking their maximum. The resulting solution uses at most $\ell$ stations and satisfies
\[
A_\ell-S\le P+2y+ax-S=(a-1)x+y.
\]
Consequently,
\[
A_\ell-S\le
\min\left\{\frac P2,\;x-y,\;(a-1)x+y\right\}.
\]
Throughout the recurrence, $1<a\le 4/3$, so the weights $2a$, $2-a$, and $2$ are positive. Applying these weights to the three bounds yields
\begin{align*}
(4+a)(A_\ell-S)
 &\le aP+(2-a)(x-y)+2\bigl((a-1)x+y\bigr)\\
 &=a(P+x+y)=aS.
\end{align*}
Therefore,
\[
A_\ell\le \left(1+\frac{a}{4+a}\right)S
          =\alpha_\ell\OPT,
\]
completing the induction.
\end{proof}

\begin{corollary}
\label{cor:finite}
The optimal ordinary Em-VC solution using at most two stations satisfies $A_2\le \frac54\OPT$. Moreover,
\[
\alpha_2=\frac54,
\qquad \alpha_3=\frac{26}{21},
\qquad \alpha_4=\frac{68}{55},
\]
and
\[
\lim_{\ell\to\infty}\alpha_\ell=\sqrt5-1.
\]
\end{corollary}

Indeed, $f(a)=1+a/(4+a)$ is increasing, with positive fixed point $\rho=\sqrt5-1$. For $a>\rho$, we have $\rho<f(a)<a$. Thus, starting at $\alpha_1=4/3$, the sequence decreases to $\rho$.

\subsection{An exact two-station algorithm and its complexity}
\label{sec:exact}

Fix distinct vertices $u,v$. For each budget $p\in D_u$ at $u$, let
\[
F_u(p)=\{e\in E:c_u(e)>p\}
\]
be the edges not covered by $u$. The minimum budget required at $v$ is
\[
q_v(p)=\max\bigl(\{c_v(e):e\in F_u(p)\}\cup\{0\}\bigr).
\]
The optimal cost for this station pair is therefore $\min_{p\in D_u}(p+q_v(p))$. Enumerating all station pairs yields $A_2$. Allowing either budget to be zero also includes single-station solutions.

\begin{algorithm}[Optimal ordinary coverage with at most two stations]
\label{alg:two}
\mbox{}\\
\textbf{Input:} A connected graph $G=(V,E)$ with positive edge lengths.\\
\textbf{Output:} At most two stations and their budgets, with total cost $A_2$.
\begin{enumerate}
\item Compute all-pairs shortest-path distances and all response costs $c_v(e)$.
\item For each vertex $u$, sort the edges in nondecreasing order of $c_u(e)$.
\item For each ordered pair $(u,v)$ of distinct vertices, compute the suffix maxima of $c_v(e)$ along the edge ordering obtained for $u$.
\item For each ordered pair $(u,v)$, enumerate $p=0$ and the distinct thresholds $c_u(e)$. Use the suffix maxima to obtain $q_v(p)$ and record the minimum value of $p+q_v(p)$.
\item Return the stations and budgets attaining the global minimum, omitting zero-budget stations.
\end{enumerate}
\end{algorithm}

\begin{theorem}
\label{thm:exact}
Algorithm~\ref{alg:two} computes $A_2$ exactly and is therefore a $5/4$-approximation algorithm for ordinary Em-VC. Given all-pairs shortest-path distances, its running time is
\[
O(nm\log m+n^2m).
\]
\end{theorem}

\begin{proof}
The budget-threshold property ensures that some optimal solution for a station pair has its first budget in $D_u$. Once this budget is fixed, the second station must cover, and need only cover, $F_u(p)$. Its minimum required budget is exactly $q_v(p)$, so the enumeration finds the restricted optimum. Corollary~\ref{cor:finite} gives the approximation guarantee.

Sorting the edges for each $u$ takes $O(m\log m)$ time. For a fixed pair $(u,v)$, computing suffix maxima and scanning all thresholds takes $O(m)$ time. Equal thresholds are processed as a group, using the suffix with response costs strictly greater than the current budget. The total sorting time is $O(nm\log m)$, and the total time for station-pair scans is $O(n^2m)$.
\end{proof}

Using Floyd--Warshall to compute all-pairs shortest-path distances takes $O(n^3)$ time, giving total time
\[
O(n^3+nm\log m+n^2m).
\]
Since the graph is connected, this also equals $O(nm\log m+n^2m)$. For any fixed $\ell$, enumerating at most $\ell$ stations and their budget thresholds, and checking every edge, computes $A_\ell$ exactly in
\[
O\bigl(n^3+\ell m\bigl(n(m+1)\bigr)^\ell\bigr)
\]
time. This yields a polynomial-time approximation algorithm with the guarantee of Theorem~\ref{thm:finite} for every fixed station limit. These bounds count arithmetic operations; for rational edge lengths encoded in binary, the computations also have polynomial bit complexity when $\ell$ is fixed.

\section{Approximation for multiple coverage with the minimum number of stations}
\label{sec:multiple}

\subsection{Selecting full-coverage stations}

By Section~\ref{sec:robustmodel}, sorting the full-coverage costs directly gives an optimal fault-tolerant solution using exactly $k$ stations.

\begin{algorithm}[Minimum-station solution for $k$-fold coverage]
\label{alg:robust}
\mbox{}\\
\textbf{Input:} A connected graph $G=(V,E)$ with positive edge lengths and an integer $1\le k\le n$.\\
\textbf{Output:} $k$ distinct stations and their budgets, with total cost $U_k$.
\begin{enumerate}
\item Compute all-pairs shortest-path distances and all response costs $c_v(e)$.
\item Compute $M(v)=\max_{e\in E}c_v(e)$ for each vertex $v$.
\item Select $k$ distinct vertices with the smallest values of $M(v)$ and give each selected vertex $v$ budget $M(v)$.
\end{enumerate}
\end{algorithm}

Every selected station covers every edge, so the output provides $k$-fold coverage. Excluding shortest-path preprocessing, computing the full-coverage costs and sorting take $O(nm)$ and $O(n\log n)$ time, respectively. With Floyd--Warshall preprocessing, the total time is $O(n^3+nm+n\log n)$.

\subsection{The approximation guarantee for general multiplicity}

\begin{theorem}
\label{thm:general}
For every $1\le k\le n$, Algorithm~\ref{alg:robust} satisfies
\[
U_k\le 2\OPT_k.
\]
\end{theorem}

\begin{proof}
Take any feasible $k$-fold coverage solution of total budget $S$. Order all vertex budgets so that $p_{u_1}\ge\cdots\ge p_{u_n}$. Corollary~\ref{cor:path} gives $kM(u_i)\le 2S-p_{u_i}$ for every $i$. Since $U_k$ is the minimum sum of full-coverage costs over $k$ distinct vertices,
\[
U_k\le\sum_{i=1}^k M(u_i)
     \le 2S-\frac1k\sum_{i=1}^k p_{u_i}
     \le 2S.
\]
Taking $S=\OPT_k$ proves the claim.
\end{proof}

\section{An improved approximation guarantee for twofold coverage}
\label{sec:twofold}

\subsection{The ball intersection graph and an independent-set bound}

We now improve the approximation ratio of Algorithm~\ref{alg:robust} to $5/3$ when $k=2$. Fix an optimal twofold-coverage solution with total budget $S=\OPT_2$. If it has exactly two positive-budget stations, each covers every edge, and hence $U_2=S$. We therefore consider the case of at least three positive-budget stations. Order the positive budgets nonincreasingly and denote the three largest by
\[
h_1\ge h_2\ge h_3=t>0.
\]

Construct the intersection graph $H$ of the positive-budget balls: each station corresponds to a vertex, and two vertices are adjacent if and only if their closed balls intersect. By the connectivity argument in Section~\ref{sec:merging}, $H$ is connected. Give vertex $v$ weight $p_v$. Let $\mathcal{I}(H)$ be the family of independent sets of $H$, and define
\[
I=\max_{J\in\mathcal{I}(H)}\sum_{v\in J}p_v.
\]
Choose a continuous center of $X$ and an original edge $e_0$ whose core path contains this center. We call $e_0$ the central edge and write $w=w(e_0)$. We establish three upper bounds on $U_2$ in terms of $S,I,w,t$.

\begin{lemma}
\label{lem:independent}
\[
U_2\le 2S-2I+w.
\]
\end{lemma}

\begin{proof}
Take an independent set of $H$ of weight $I$. Its balls are pairwise disjoint, so each point of $X$ belongs to at most one of them. Deleting these balls from the twofold cover leaves a cover of $X$ with total budget $S-I$. Ball merging and center rounding give
\[
C\le S-I,
\qquad U_2\le 2C+w\le 2S-2I+w.
\]
\end{proof}

\subsection{Budget estimates along induced paths}

Assign length $p_i+p_j$ to each edge $ij$ of $H$. Every shortest path in this weighted graph is induced. Indeed, if a chord joined two nonconsecutive vertices on the path, replacing the intervening subpath by that chord would decrease the length by twice the sum of the budgets of the internal vertices. All vertex weights are positive, so this would strictly shorten the path, a contradiction.

\begin{lemma}
\label{lem:induced}
For every positive-budget station $u$,
\[
M(u)\le 4I-p_u.
\]
Consequently,
\[
U_2\le 8I-h_1-h_2.
\]
\end{lemma}

\begin{proof}
Take distinct positive-budget stations $u,v$ and a shortest path between them in the weighted intersection graph. Let $W$ be the sum of its vertex budgets. Since the path is induced, the vertices in odd positions and those in even positions each form an independent set. Thus, $W\le 2I$. The distance between the centers of two intersecting balls is at most the sum of their radii. The triangle inequality therefore gives
\[
d_X(u,v)\le 2W-p_u-p_v\le 4I-p_u-p_v.
\]
For every edge $e$, twofold coverage provides a station $v\ne u$ that covers it. Lemma~\ref{lem:equivalence} gives
\[
c_u(e)\le d_X(u,v)+c_v(e)
        \le 4I-p_u-p_v+p_v=4I-p_u.
\]
Taking the maximum over all edges proves the first claim. Choosing the two largest-budget stations of the original solution as full-coverage stations gives $U_2\le 8I-h_1-h_2$.
\end{proof}

\subsection{Savings from merging along the central edge}

\begin{lemma}[Geometry of the central edge]
\label{lem:central}
Suppose stations $a,b$ both cover $e_0$ and have budgets $A,B$, respectively. Then
\[
A\ge w,
\qquad B\ge w,
\qquad d_X(a,b)\le A+B-w.
\]
Their balls can be merged into a ball of radius at most $A+B-w/2$.
\end{lemma}

\begin{proof}
For each of $a$ and $b$, choose a nearer endpoint of $e_0$. Since both stations cover the edge, their distances to their chosen endpoints are at most $A-w$ and $B-w$, respectively. If the endpoints coincide, the center distance is at most $A+B-2w$. Otherwise, connecting the two endpoints along the core path of the central edge gives a center distance of at most $A+B-w$. Both cases yield the stated bound.

By Lemma~\ref{lem:merge}, the merged radius is at most
\[
\max\left\{A,B,\frac{A+B+d_X(a,b)}2\right\}
\le A+B-\frac w2,
\]
where we use $A,B\ge w$.
\end{proof}

\begin{lemma}
\label{lem:thirdbudget}
\[
U_2\le 2S-2t.
\]
\end{lemma}

\begin{proof}
Choose any two stations covering $e_0$, and call them $a,b$. Among the three largest-budget stations, at least one station $r$ is outside $\{a,b\}$, and $p_r\ge t$. Deleting its ball leaves an ordinary cover containing the balls of $a$ and $b$. First merge these two balls using Lemma~\ref{lem:central}, decreasing the total radius by at least $w/2$. Then merge the remaining intersecting balls until one ball covers $X$. Consequently,
\[
C\le S-p_r-\frac w2\le S-t-\frac w2.
\]
Together with Lemma~\ref{lem:round}, this gives
\[
U_2\le 2C+w\le 2S-2t.
\]
\end{proof}

\subsection{The third upper bound}

\begin{lemma}
\label{lem:thirdbound}
\[
U_2\le 8I-4w+2t.
\]
\end{lemma}

\begin{proof}
We distinguish two cases according to the central edge length $w$ and the third-largest budget $t$.

\noindent\textit{Case 1: $w\le t$.}
At least two stations cover the central edge, and each has budget at least $w$. Hence, $h_1,h_2\ge w$. By Lemma~\ref{lem:induced},
\[
U_2\le 8I-h_1-h_2\le 8I-2w\le 8I-4w+2t.
\]

\noindent\textit{Case 2: $w>t$.}
No station ranked third or lower by budget can cover the central edge. Therefore, its two covering stations must be the two largest-budget stations. Denote them by $a,b$, with budgets $A=h_1$ and $B=h_2$. Then
\[
A\ge B\ge w>t.
\]
We first bound $M(a)$. Consider an edge $e$ not covered by $a$. Twofold coverage provides a covering station $v\notin\{a,b\}$. In the intersection graph with edge lengths $p_i+p_j$, take a shortest path from $a$ to $v$, written as
\[
a=v_0,\qquad j=v_1,\qquad v_2,\ldots,v_r=v.
\]
This path is induced, so $a$ is not adjacent to any of $v_2,\ldots,v_r$. The even-indexed tail vertices together with $a$ form an independent set, and so do the odd-indexed tail vertices together with $a$. Hence,
\[
\sum_{i=2}^r p_{v_i}\le 2(I-A).
\]
If the path has only one edge, the tail is empty and the sum is zero. Let $W$ be the sum of all vertex budgets on the path. Then
\[
W\le A+p_j+2(I-A)=2I-A+p_j.
\]
The sum of the path's edge lengths is $2W-A-p_v$. Since $c_v(e)\le p_v$,
\[
c_a(e)\le d_X(a,v)+c_v(e)
        \le 2W-A\le 4I-3A+2p_j.
\]
If $j\ne b$, then $p_j\le t$, giving $c_a(e)\le 4I-3A+2t$. If $j=b$, Lemma~\ref{lem:central} bounds the actual center distance in the first step by $A+B-w$, which saves $w$ compared with its intersection-graph edge length $A+B$. Using the usual intersecting-ball distance bound on all subsequent steps gives
\[
c_a(e)\le 4I-3A+2B-w.
\]
Combining the two possibilities, every edge not covered by $a$ satisfies
\[
c_a(e)\le 4I-3A+2\max\left\{t,B-\frac w2\right\}.
\]
Edges already covered by $a$ satisfy $c_a(e)\le A$. Since the singleton $\{a\}$ is independent, $I\ge A$, and the right-hand side of the preceding bound is also at least $A$. Taking the maximum over all edges yields
\[
M(a)\le 4I-3A+2\max\left\{t,B-\frac w2\right\}.
\]
Interchanging $a$ and $b$ gives
\[
M(b)\le 4I-3B+2\max\left\{t,A-\frac w2\right\}.
\]
For any $q\ge w>t$,
\[
2\max\left\{t,q-\frac w2\right\}\le 3q-2w+t.
\]
Indeed, if the maximum is $t$, subtracting $2t$ from the right-hand side gives $3q-2w-t\ge w-t>0$. If the maximum is $q-w/2$, subtracting $2q-w$ gives $q-w+t\ge0$. Apply this inequality with $q=B$ and $q=A$ to the two full-coverage cost bounds and add them:
\begin{align*}
M(a)+M(b)
 &\le 8I-3A-3B+(3B-2w+t)+(3A-2w+t)\\
 &=8I-4w+2t.
\end{align*}
Since $U_2\le M(a)+M(b)$, the claimed bound holds in both cases.
\end{proof}

\subsection{Combining the three bounds}

\begin{theorem}[Approximation for twofold coverage]
\label{thm:twofold}
For $k=2$, Algorithm~\ref{alg:robust} satisfies
\[
U_2\le\frac53\OPT_2.
\]
\end{theorem}

\begin{proof}
If an optimal solution has only two positive-budget stations, then $U_2=\OPT_2$. Otherwise, Lemmas~\ref{lem:independent}, \ref{lem:thirdbudget}, and~\ref{lem:thirdbound} give
\begin{align*}
U_2&\le 2S-2I+w,\\
U_2&\le 2S-2t,\\
U_2&\le 8I-4w+2t.
\end{align*}
Weight these bounds by $4,1,1$, respectively. The quantities $I,w,t$ cancel simultaneously, yielding
\begin{align*}
6U_2
 &\le 4(2S-2I+w)+(2S-2t)+(8I-4w+2t)\\
 &=10S.
\end{align*}
Since $S=\OPT_2$, the result follows.
\end{proof}

\subsection{A lower-bound instance for two-station solutions}

Consider a triangle with unit edge lengths. Assigning budget $1$ to each of its three vertices gives a twofold cover of total cost $3$. After removing stations that cover no edges, every remaining station has budget at least $1$, so a three-station solution costs at least $3$. In any two-station solution, each station must cover the opposite edge and therefore needs budget at least $2$. Thus,
\[
\OPT_2=3,
\qquad U_2=4,
\qquad \frac{U_2}{\OPT_2}=\frac43.
\]
Together with Theorem~\ref{thm:twofold}, this places the worst-case approximation ratio of the optimal two-station algorithm for twofold coverage in $[4/3,5/3]$.

\section{Conclusion and future directions}
\label{sec:conclusion}

Allowing two stations improves the approximation guarantee for ordinary Em-VC from $4/3$ to $5/4$, and our recurrence gives further bounds for each fixed station limit. For $k$-fold coverage, the best solution using exactly $k$ stations is a $2$-approximation, with an improved guarantee of $5/3$ when $k=2$. It remains open whether the bounds from the finite-station recurrence are tight. For twofold coverage, closing the gap between the lower bound of $4/3$ and the upper bound of $5/3$ for the best two-station solution is another open question.

\section*{Declaration of generative AI and AI-assisted technologies in the manuscript preparation process}

The authors used an AI assistant for translation and language revision, manuscript organization, mathematical discussion, and reference checking. The authors reviewed and edited the content as needed and assume full responsibility for all content.

\end{document}